\documentclass[a4paper]{article}
\usepackage[utf8]{inputenc}
\usepackage{amsmath,amssymb,amsthm,mathtools,comment}
\usepackage{authblk}
\usepackage{todonotes}
\usepackage{url}
\usepackage{hyperref}
\usepackage[margin=1in]{geometry}
\usepackage{listings}
\usepackage{multirow}
\usepackage{realhats}
\usepackage{caption}
\usepackage{physics}
\usepackage{cite}
\usepackage{subcaption, graphicx}
\usepackage{color}
\usepackage[nottoc]{tocbibind}
\usepackage{footmisc}

\newtheorem{theorem}{Theorem}

\newcommand{\ii}{\mathrm i}

\title{Hamiltonian-Oriented Homotopy QAOA}
\author{Akash Kundu$^{1,2}$\footnote{A. Kundu and L. Botelho contributed equally to this work}, Ludmila Botelho$^{1,2 *}$\thanks{corresponding author, lbotelho@iitis.pl}, Adam Glos$^{1,3}$}
\affil{$^1$Institute of Theoretical and Applied Informatics, Polish Academy of Sciences, Ba\l tycka~$5$, Gliwice, Poland}
\affil{$^2$Joint Doctoral School, Silesian University of Technology,\\Akademicka $2$A, Gliwice, Poland}
\affil{$^3$Algorithmiq Ltd, Kanavakatu 3C 00160 Helsinki, Finland}
\date{}
\begin{document}
\maketitle
\begin{abstract}
	The classical homotopy optimization approach has the potential to deal with highly nonlinear landscape, such as the energy landscape of QAOA problems. Following this motivation, we introduce Hamiltonian-Oriented Homotopy 
	QAOA (HOHo-QAOA), that is a heuristic method for combinatorial optimization using QAOA, based on classical homotopy optimization. The method consists of a homotopy map that produces an optimization problem for each value of interpolating parameter. Therefore, HOHo-QAOA decomposes the optimization of QAOA into several loops, each using a mixture of the mixer and the objective Hamiltonian for cost function evaluation. Furthermore, we conclude that the HOHo-QAOA improves the search for low energy states in the nonlinear energy landscape and outperforms other variants of QAOA.
\end{abstract}

\section{Introduction}
Speedup of practical applications is yet to be realized for quantum devices as they are small and noise prone. The limitations of available hardware initiated the Noisy Intermediate Scale Quantum (NISQ) era \cite{preskill2018quantum}. The NISQ algorithms \cite{bharti2022noisy} can operate on limited amount of resources., in particular  by  distributing tasks between quantum and classical devices.
Many of those algorithms are represented by a broad class of variational quantum algorithms (VQAs)~\cite{cerezo2021variational}.
Their generic structure consists of two subroutines: a parametric quantum circuit (PQC) is implemented on quantum hardware that generates a quantum state, and classical hardware calculates the cost function and optimizes the parameters of PQC.
One of the advantages of VQAs is that they can be easily adapted to various computational problems as long as the Hamiltonian can be designed whose ground state corresponds to the solution of the problem. To mention a few, VQAs has potential applications in finding the ground state of a molecule~\cite{kandala2017hardware}, solving linear~\cite{bravo2019variational} and nonlinear~\cite{lubasch2020variational} system of equations, quantum state-diagonalization~\cite{larose2019variational}, and quantum device certification~\cite{kundu2022variational}. A detailed review can be found in~\cite{cerezo2021variational}.

Quantum approximate optimization algorithm (QAOA)~\cite{farhi2014quantum} is a variational quantum algorithm dedicated to combinatorial optimization problems. The PQC in QAOA is a trotterized adiabatic evolution i.e. the circuit consist of interchangeably applied so-called mixer and problem Hamiltonians. It has potential application in solving problems like graph coloring \cite{farhi2001quantum, tabi2020quantum, bako2022near}, MaxE3Lin2~\cite{farhi2014maxkcut}, Max-$K$-Vertex Cover~\cite{cook2020quantum}, or traveling salesman problem~\cite{glos2022space, bako2022near}. To improve the performance of QAOA, multiple optimization strategies have been introduced~\cite{medvidovic2021classical,wang2018quantum,alam2020accelerating,shaydulin2019multistart,hegade2021portfolio,zhu2022adaptive,zhouprx2020, zhou2022qaoa, patel2022rlrqaoa}. This is becuase given the limited resources of quantum computers it is essential to effectively explore the landscape of cost function for PQC. On the othe hand, the landscape of energy function in QAOA is highly nonlinear and to deal with such complicated landscapes, sophisticated methods are necessary. 

This motivate us to formulate a heuristic optimization strategy that uses classical homotopy optimization for  QAOA. The homotopy optimization has potential application in dealing with highly nonlinear functions \cite{layne1989optimization}. 
The homotopy method comprises a homotopy map, which for each value of interpolating parameter $\alpha\in[0,1]$  outputs an optimization problem. In particular, for $\alpha=0$, the problem is easy-to-solve, and for $\alpha=1$ the homotopy map returns the problem of interest. During the interpolation process, which changes the value of $\alpha$ from 0 to 1, the solution continuously changes and is expected to be optimal, or close to optimal for the intermediate problems.  If the intermediate optimization succeed, in the end we obtain the optimum of the target problem. One can see quantum annealing as a particular type of homotopy optimization. A homotopy optimization for VQE was already proposed in~\cite{garcia2018addressing} and improved in~\cite{mcclean2016theory, harwood2022improving}. However, its applicability for QAOA was only briefly mentioned in~\cite{mcclean2021low}.

The introduced Hamiltonian-Oriented Homotopy QAOA (HOHo-QAOA) decomposes the optimization into several loops. The homotopy map smoothens between the mixer Hamiltonian and the problem Hamiltonian during the optimization and each loop uses the mixture of these two Hamiltonians for cost function evaluation. In each loop, the quantum state is optimized with respect to such intermediate cost functions. This strategy simplifies the search for good QAOA parameters while keeping the PQC unchanged. To show this, first we empirically analyze the impact of the choice of the homotopy parameters: the initial $\alpha_\text{init}$ value and the step parameter $\alpha_\textrm{step}$ which defines the difference between two consecutive $\alpha$ values. Although theoretically, a choice of $\alpha_\text{init}$ and $\alpha_\text{step}$ very close to zero provides a better approximation to the optimal solution, empirically we show that one can still get a good approximation to the optimal solution even if $\alpha_\text{init}$ and $\alpha_\text{step}$ are detached from zero. This hugely reduces the computational cost of HOHo-QAOA. Finally, we compare HOHo-QAOA with other commonly used QAOA optimization strategies~\cite{farhi2014quantum,zhouprx2020}.

The rest of the paper is organized in the following way. In Section~\ref{sec:prelims}, we provide a brief overview of the adiabatic quantum computing, variants of QAOA and the homotopy method. Throughout Section~\ref{sec:nc-qaoa}, we numerically investigate the efficient settings of the homotopy parameters. Furthermore, we compare HOHo-QAOA with the other variants of QAOA considered in the literature. Finally, we conclude the article in Section~\ref{sec:conclusion}.
	
\section{Preliminaries}\label{sec:prelims}
\subsection{QAOA}\label{sec:aqc}
The core concept of Adiabatic Quantum Computing (AQC) lies in the adiabatic theorem. Let us consider $H(s) = H(t/T)$, a time-dependent smoothly varying Hamiltonian for all $t\in[0,T]$ i.e. $s\in[0,1]$, where $T$ is the total time of evolution. Let us denote by   $\ket{E_i(s)}$ an eigenvector of $H(s)$ with corresponding eigenvalue $E_i(s)$, where we assume $E_0(s)\leq E_1(s)\leq\ldots$. The adiabatic theorem roughly states that a system that is initially prepared in $\ket{E_0(0)}$ of $H(t=0)$, after time-evolution that is piloted by Schr\"odinger equation with the given Hamiltonian $H(s)$, will approximately keep the state of the system in the $\ket{E_0(1)}$ at $t=T$, provided that the change in $H(s)$ is ``sufficiently slow". Traditionally the sufficiently slow change is given by the condition~\cite{messiah1961quantum, albash2018adiabatic}
\begin{equation}
	T\gg\Delta^{-2}\underset{s\in[0,1]}{\text{max}}\norm{ \left[\frac{\partial H(s)}{\partial t}\right]^2 },\label{eq:aqc-condition}
\end{equation}
where $\Delta=\text{min}_s\left(E_1(s)-E_0(s)\right)$, is the spectral gap.
A class of independent conditions on $T$ has been discussed in \cite{marzlin2004inconsistency, tong2005quantitative, du2008experimental,wu2008adiabatic}. AQC has the potential to take a initial Hamiltonian say $H_\text{mix}$ whose ground state is easy-to-prepare to the ground state of a computationally hard problem Hamiltonian $H_\text{obj}$. A particular time-dependent Hamiltonian interpolates between the $H_\text{mix}$ and $H_\text{obj}$ as 
\begin{equation}
	H(s) = \left( 1-s\right)H_\textrm{mix} +  sH_\textrm{obj},
	\label{eq:qa}
\end{equation}
AQC in the form of quantum annealing has been used for a variety of applications including real-world problems~\cite{salehi2022unconstrained, domino2021quantum, domino2022quadratic, borowski2020new, Arya2022, glos2022testvehicle}, and in quantum chemistry~\cite{babbush2014adiabatic}. For a rigorous review of AQC check~\cite{dascolloquium-qa-aqc, albash2018adiabatic}.

The Quantum Approximate Optimization Algorithm (QAOA), uses the first order Suzuki-Trotter transformation of  $\textrm{exp} ( -\ii H(s))$ as the variational ansatz to solve combinatorial optimization problems. The trotterization gives rise to the operators $\textrm{exp} ( -\ii\gamma_j H_\textrm{obj})$ and $\textrm{exp}( -\ii\beta_j H_\textrm{mix})$, where $\gamma_j$ is the parameter corresponding to objective Hamiltonian and $\beta_j$ corresponds to mixer Hamiltonian for $j$-th step. The mixer Hamiltonian is traditionally expressed as $H_\textrm{mix} = -\sum_iX_i$, where $X_i$ is Pauli $X$ operator acting on $i$-th qubit and $H_\textrm{obj}$ is the objective Ising Hamiltonian whose ground state encodes the optimal solution of the problem. This results in state
	\begin{equation}
	\lvert \vec{\gamma},\; \vec{\beta}\rangle = \prod_{j=1}^{L} \textrm{exp}\left( -\ii\beta_j H_\textrm{mix}\right) \textrm{exp}\left( -\ii\gamma_j H_\textrm{obj}\right)\lvert+\rangle^{\otimes N}\label{eq:qaoa-suzuki-trotter},
\end{equation}
where $N$ is the number of qubits, $L$ is the number layers that defines the number of repeated application of mixer and objective Hamiltonian, and $\lvert+\rangle^{\otimes N}$ is the ground state of $-\sum_iX_i$. The algorithm utilizes quantum hardware to evaluate the energy expectation value $E( \vec{\gamma}, \vec{\beta} ) = \langle \vec{\gamma},\vec{\beta}\lvert H_\textrm{obj} \rvert \vec{\gamma},\vec{\beta} \rangle$. Then the parameters $\vec{\gamma}$ and $\vec{\beta}$ are optimized using classical optimization methods so that the energy is minimized. This energy evaluation along with classical optimization  QAOA is well defined for any combinatorial optimization problems as long as $H_\textrm{obj}$ can be implemented efficiently. While the proposed $X$-mixer combined with 2-local Ising model is frequently used in the literature, different choices were also considered \cite{hadfield2019quantum,bako2022near, bartschi2020grover,wang2020x,glos2022space}. 

Heuristic learning of QAOA has been explored in trajectories QAOA (T-QAOA)~\cite{zhouprx2020}. T-QAOA is a heuristic strategy that utilizes interpolation-based prediction of good QAOA parameters.
With the random initialization, the cost for optimization of QAOA is exponential in the number of layers of QAOA~\cite{zhouprx2020}.  On the other hand, with increased number of layers, $H_\text{mix}$ may gradually turn off while the $H_\text{obj}$ turns on, which is reminiscent of AQC. However, QAOA could learn via following a diabatic path to achieve higher success probability \cite{crosson2014different, PhysRevX.6.031010, RevModPhys.90.015002}, which is beyond the adiabatic process natural for AQC. This fact was used in T-QAOA by reusing the optimal angles found for $L$-layers in the $(L+1)$-layers PQC. 

The T-QAOA variant considered in this paper runs as follow. It starts with a number of layer $L_0$, and finds the locally optimal parameters $(\vec{\gamma}^{L_0}, \vec{\beta}^{L_0})$. Then it uses the optimal parameters of layer $L_0$ to construct the initial parameters for the layer $L_0+1$ by sampling the last entries of $\vec{\gamma}^{L_0+1}$ from a uniform random distribution and setting  $\vec{\beta}^{L_0+1}=0$. With such initialization, the $(L_0+1)$-th layer PQC is optimized, and the procedure is repeated until a final number of layer $L$ is reached. Note that different interpolation method can be used~\cite{zhouprx2020}.

Note that for QAOA, the energy landscape with respect to a single parameter $\theta$ is related to the following process. First, an initial quantum state is prepared. Then, if applicable, all the unitary operations which precedes the $\theta$-dependent operation are applied, which transforms the
\begin{figure}[tbh!]
	\centering
	\includegraphics{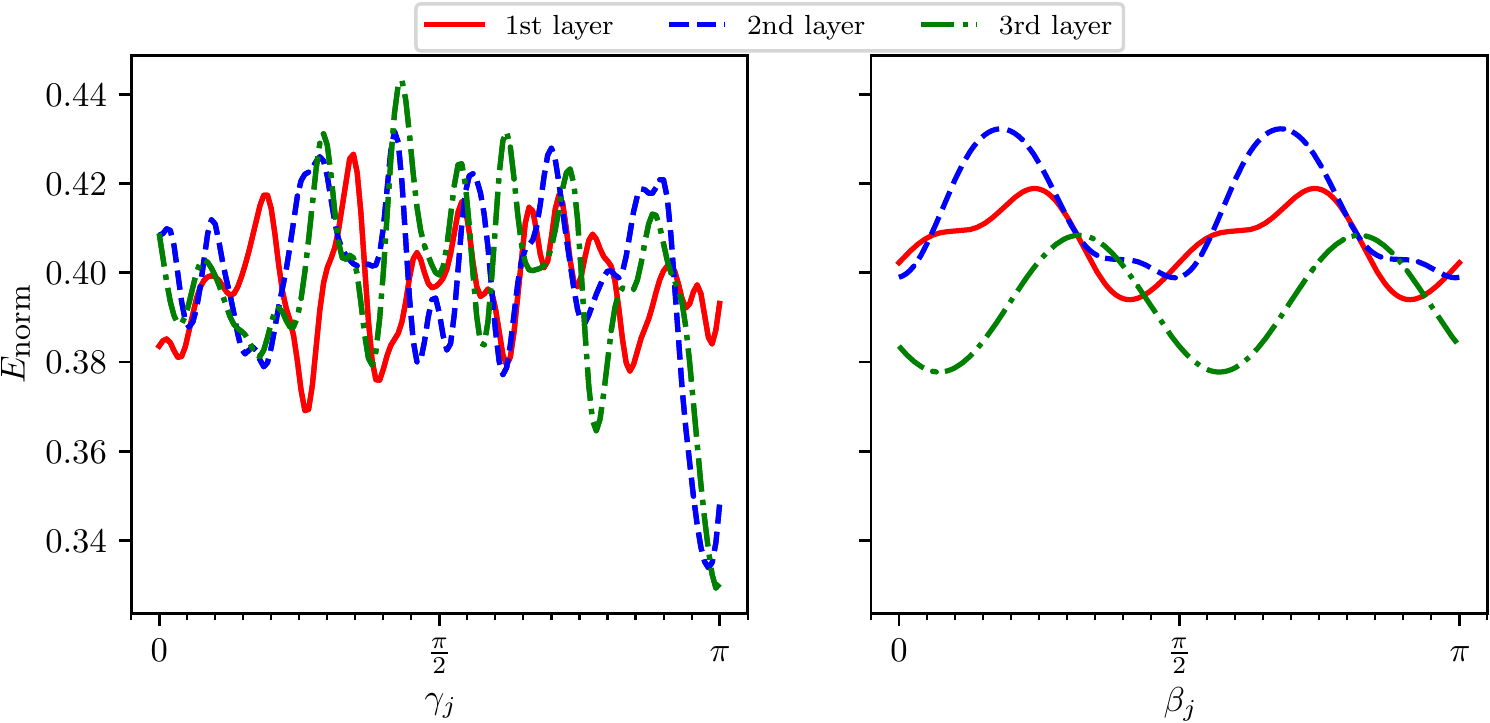}
	\caption{Illustration of highly nonlinear energy landscape of QAOA for Max-Cut for $10$ nodes with weighted Barab\'asi-Albert graph for objective Hamiltonian (left) and mixer Hamiltonian (right). $E_{norm}$ is a standarized energy of the objective Hamiltonian, so that the eigenvalues lies in $[0,1]$ }
	\label{fig:non-linear_energy_landscape}
\end{figure}
initial state into a different state (possibly a mixed state for noisy evolution). Afterwards, under an assumption of pure evolution, a unitary $\exp(- \ii \theta H)$ for mixer or objective Hamiltonian $H$ is applied. Finally, the remaining operations are applied and the energy estimation with respect to observable is conducted. As shown in the Appendix~\ref{apndx:proof-non-linear-lanscape}, the energy function with respect to $\theta$ takes the form
\begin{equation}
	C+ \sum_{i>j} A_{i,j}\cos(\theta(E_i-E_j) + B_{i,j})\label{eq:non-linear-landscape},
\end{equation}
in which $\{E_i\}$ is the set of all eigenvalues of the operator $H$, and real parameters $C,A_{i,j}, B_{i,j}$ depend on the initial state, observable, and $\theta$-independent quantum operations. Note that Eq.~\eqref{eq:non-linear-landscape} is
highly nonlinear, therefore its optimization may be difficult in practice. This is in contrast to typically used VQE approaches, in which the parameter-dependent unitary can be reduced to a single-qubit gate, which in turn may result in a simple, yet powerful gradient-free optimization technique~\cite{ostaszewski2021structure, nakanishi2020sequential}.

Unfortunately, the number of cosines in Eq.~\eqref{eq:non-linear-landscape} may grow quadratically with number of distinct eigenvalues of the considered Hamiltonian. In the case of the objective Hamiltonian the number may be particularly high. While for many simple problems like unweighted Max-Cut or Max-SAT the number of different eigenvalues usually grows polynomially with the size of the data, for weighted Max-Cut each partition may result in a different objective value, which may give $\order{2^n}$ different energies in general. A complicated energy landscape can be seen already even for a small and simple instance, see Fig.~\ref{fig:non-linear_energy_landscape}. For problems generating such a complicated landscapes, more sophisticated methods may be at hand.

\subsection{Homotopy optimization method}
One of the well-known methods to solve a system of highly nonlinear problems is \textit{homotopy optimization}, where a homotopy map is constructed between two systems. The solution corresponding to one of the systems is transformed into the solution of the other system.  For example, consider the function $f_\text{targ}(x)$ which encodes a computationally hard problem and $f_\text{init}(x)$ which is a problem with an easy-to-find solution. Then the particular homotopy map between the systems can be given as
\begin{equation}
	\mathcal{F}(\alpha, x) = g_1(\alpha)f_\text{targ}(x) + g_2(\alpha)f_\text{init}(x),\;\;\;\;\; 0\leq\alpha\leq1,
\end{equation}
where  
\begin{equation}
	\begin{split}
	&g_1(0) = 0,\quad g_2(0) = 1,\\ 
	&g_1(1) = 1,\quad g_2(1) = 0.	
	\end{split}
\end{equation}
Here, we get a family of problems corresponding to $\min_x\mathcal{F}(\alpha, x)=0$ for each $\alpha$ value from $0$ to $1$. We track the optimized solutions starting from $\left(\alpha, x\right)$ = $\left(0, x_0\right)$, as $\alpha$ moves from $0$ to $1$, which for a successful homotopy map leads to $\left(\alpha, x\right)$ = $\left(1, x_1\right)$, where $x_1$ is ideally the optimal solution of $f_\text{targ}$.

The state-of-the-art approach is to start from $\left(\alpha_\text{init}, x_\text{init}\right)$ with $x_{\rm init}$ minimizing $\mathcal{F}(0,x) = f_{\rm init}(x)$. Then the problem $\min_x\mathcal{F}(\alpha +\alpha_\text{step}, x)=0$ is iteratively solved using the solution of $\min_x\mathcal{F}(\alpha,x)$ as a starting point, for sufficiently small $\alpha_\text{step}>0$~\cite{layne1989optimization}. 

\section{ Hamiltonian-Oriented Homotopy QAOA}\label{sec:nc-qaoa}
\subsection{Proposed method}
The Hamiltonian-oriented homotopy QAOA decomposes the optimization process of the objective Hamiltonian into several optimization loops. Each loop optimizes the energy
\begin{equation}
	E_{\alpha}( \vec{\gamma}, \vec{\beta} ) = \langle \vec{\gamma},\vec{\beta}\lvert H(\alpha) \rvert \vec{\gamma},\vec{\beta} \rangle,\label{eq:mc-qaoa-energy-expectation}
\end{equation}
where $H(\alpha)$ encodes the homotopy map
\begin{equation}
	H(\alpha) = g_1(\alpha)H_\text{mix}+g_2(\alpha) H_\text{obj},\;\;\; 0\le\alpha\le1 \label{eq:nc-aqc}.
\end{equation}
For $\alpha=1$ the expectation value in Eq.~\eqref{eq:mc-qaoa-energy-expectation} is the energy corresponding to the $H_\text{obj}$. While there is a freedom in the choice of $g_1$ and $g_2$, throughout the paper we a simple case 
\begin{equation}
	g_1(\alpha) = 1-\alpha,\;\;\; g_2(\alpha)=\alpha.
\end{equation}
During the optimization process, we choose an initialization of mixer and objective parameters (at $\alpha=0$) in such a way that the parameters corresponding to the mixer are sampled from the uniform random distribution $\text{U}(a,b)$ in an interval $[a=0, b=2\pi]$ and the objective parameters are all set to $0$. With this initialization we make sure that the homotopy starts from the exact ground state of the mixer on a noise-free setting, as application of mixer on its eigenstate does not change the state. For $\alpha'>\alpha\geq0$ the initial parameters are chosen as
\begin{equation}
	( \vec{\gamma},\vec{\beta} )^\text{init}_{\alpha^\prime} = ( \vec{\gamma},\vec{\beta} )^{\ast}_\alpha,
\end{equation}
here $\ast$ denotes the optimal parameters for $\alpha$. 

It should be noted that each run of HOHo-QAOA follows the generic structure of homotopy process as in Eq.~\eqref{eq:nc-aqc} where the ``run-time" of HOHo-QAOA is characterized by the $\alpha_\text{step}$, for a fixed $\alpha_\text{init}$. The parameter $\alpha_\text{init}$ fixes the initial $\alpha$ value. Generally, it can be inferred that better approximation to the optimal solution can be achieved if we choose a sufficiently small value of $\alpha_\text{step}$ and $\alpha_\text{init}$. They can be described in a more elaborated way as follows. Small value of $\alpha_\text{step}$ helps us realizing the homotopy of Eq.~\eqref{eq:nc-aqc} and at the same time if we initiate with $\alpha_\text{init}\rightarrow0$, it becomes easier to find the ground state for the first step. To show this, throughout the paper, we investigate the normalized energy
\begin{equation}
	E_\text{norm}(E_\alpha(\vec \gamma,\vec \beta),\alpha) = \frac{E_\alpha(\vec \gamma,\vec \beta) - \text{min}H(\alpha)}{\text{max}H(\alpha)- \text{min}H(\alpha)}, \label{eq:energy_norm}
\end{equation}
with respect to parameters of HOHo-QAOA, where $E_\text{norm}(\alpha) = 0$, is the normalized ground energy for any $\alpha\in[0,1]$, and $\min H(\alpha)$ ($\max H(\alpha)$) denotes minimum (maximum) of $H(\alpha)$.

\subsection{Initialization strategy}\label{subsec:init_strategy}
In the following, first we numerically discuss proposed settings for initial QAOA parameters $( \vec \gamma,\vec\beta)^\text{init}$. With this setting we show that the homotopy parameters i.e. $\alpha_\text{init}, \; \alpha_\text{step}$ can be chosen detached from zero without compromising the efficiency of the method. We consider and optimized energy $E^\ast _{\rm norm}$, or in the case of HOHo-QAOA also an intermediate optimized step energy $E^{\ast}_{\rm norm}(\alpha)$.  In the numerical results the $E^{\ast}_\text{norm}$ is averaged over $100$ experiments. Details of the experiment can be found in Appendix~\ref{app:ex-details}.

For the numerical investigation of optimal QAOA parameters, which is illustrated in Figure~\ref{fig:initializtion-comparison}, we consider three possible initialization choices of the mixer and objective parameters at $\alpha=\alpha_{\rm init}$:
\begin{enumerate}
\item
RR (Random Random): When the parameters corresponding to mixer and objective  Hamiltonians are chosen from a uniform random distribution $\text{U}(0,2\pi)$ i.e. 
$
\gamma_{j}^\text{init} \sim \text{U}(0,2\pi),\;  \beta_{j}^\text{init} \sim \text{U}(0,2\pi). \nonumber
$
\item NZR (Near-Zero Random): The parameters corresponding to mixer  Hamiltonian are chosen from $\text{U}(0,2\pi)$ but objective parameters are sampled from the values very close zero i.e. 
$
\gamma_{j}^\text{init} \sim \text{U}(0,v), \; \beta_{j}^\text{init} \sim \text{U}(0,2\pi),\nonumber
$
where $v$ is $0.05$. 
\item 
ZR (Zero Random): Mixer parameters are sampled from $\text{U}(0,2\pi)$ chosen and objective is all zeros i.e. 
$
\gamma_{j}^\text{init} = 0, \; \beta_{j}^\text{init} \sim \text{U}(0,2\pi)\nonumber
$ as proposed before.
\end{enumerate}
\begin{figure}[t!]
	\centering
	\includegraphics{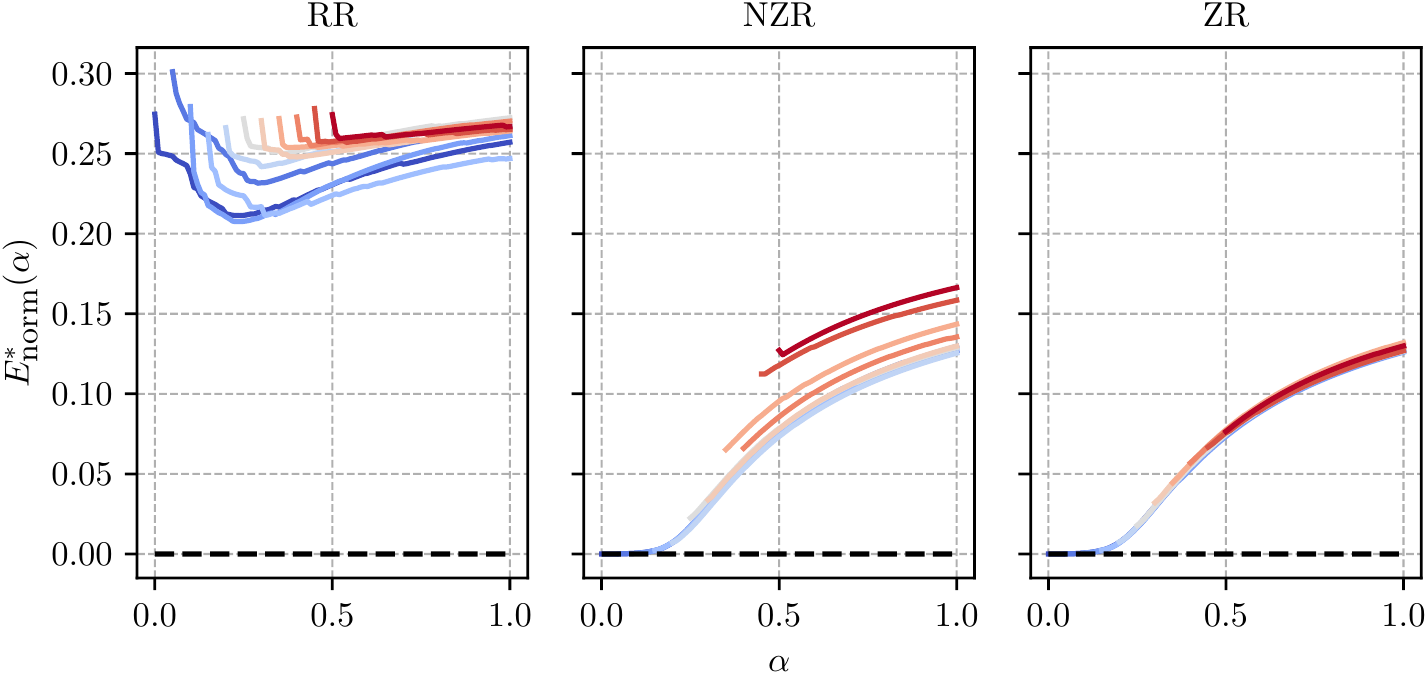}
	\caption{The impact of different methods of initialization of $\gamma_j,\beta_j$ on HOHo-QAOA. The {left}, the {middle} and the {right} figures are representing the convergence for RR (Random Random), NZR (Near-Zero Random) with parameter $v=0.05$, and ZR (Zero Random) initialization respectively, see Sec.~\ref{subsec:init_strategy} for details. It is visible that the {ZR} is outperforming the other two initializations. Although for $\alpha_\text{init} \leq 0.2$ the performance of NZR and ZR are comparable but as we tune $\alpha_\text{init} > 0.2$, the minima for {NZR} scatters in region $0.10<E_\text{norm}<0.15$ whereas the minima for {ZR} clusters in a very narrow $E_\text{norm}$-width.}
	\label{fig:initializtion-comparison}
\end{figure}
\begin{figure}[tbh!]
	\centering
	\includegraphics{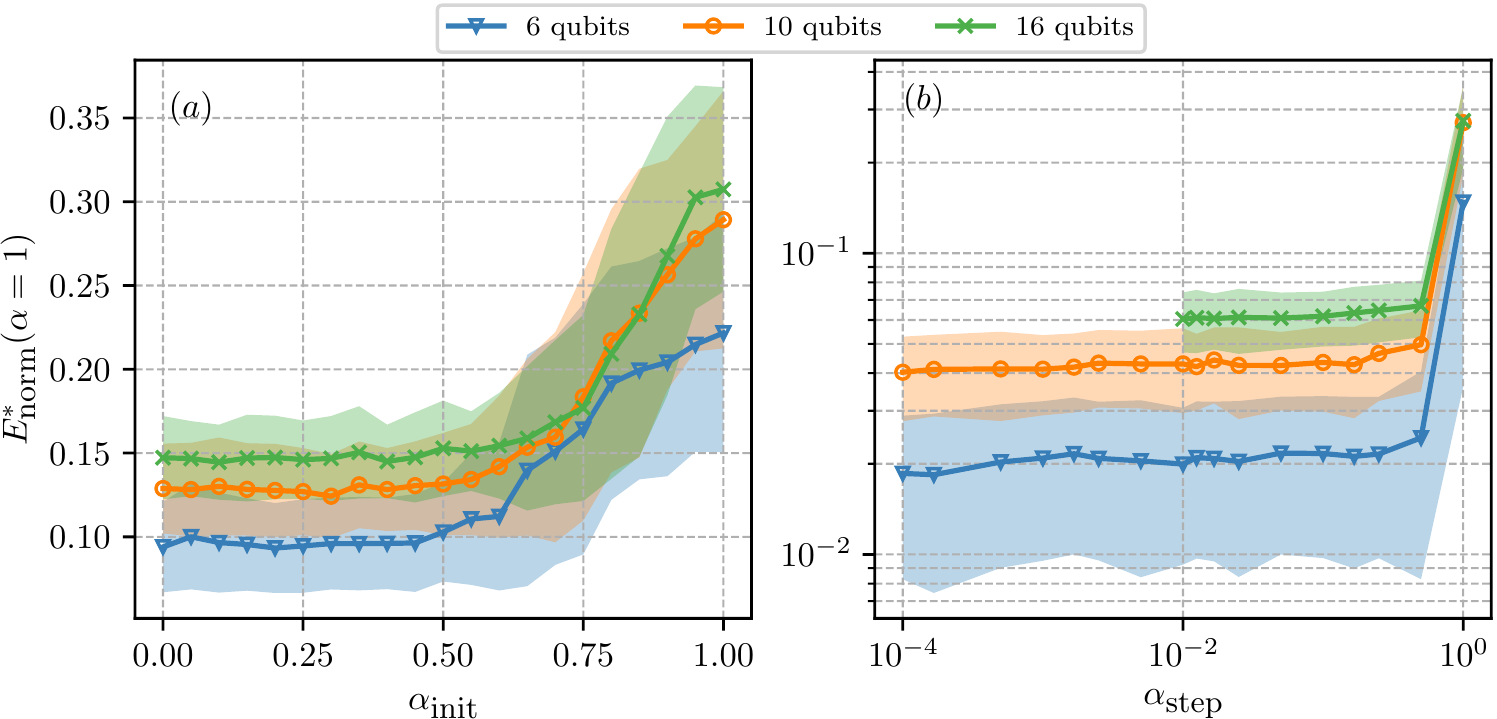}
	\caption{We illustrate the dependency of $E^{\ast}_\textrm{norm}$ with $\alpha_\textrm{init}$ and $\alpha_\textrm{step}$. In (a) the variation of $ E^{\ast}_\text{norm}$ with $ \alpha_\text{init}$ for $3$ layers of HOHo-QAOA is presented, with $\gamma_j^\text{init} = 0, \beta_j^\text{init}\sim\text{U}(0,2\pi)$. In the figure, we see a \textit{region of stability} of HOHo-QAOA in respect with $ \alpha_\text{init}$ in the range $0.0$ to $0.50$. In (b) we present $E^{\ast}_\text{norm}$ vs $\alpha_\text{step}$ using $10$ layers of HOHo-QAOA. Just like in the case of $\alpha_\textrm{init}$, for $\alpha_\textrm{step}$ a same \textit{region of stability} can be observed. This gives us the preference on the choice of \textit{step parameter} while utilizing HOHo-QAOA. It should be noted that the y-axis in (a) is in \texttt{linear} scale and whereas in (b) it is in \texttt{log} scale. The lines in both the plots are taken $\alpha_\textrm{init}$ and $\alpha_\textrm{step}$-wise and is the mean of 100 experiments. The area under the plots are standard deviation of energies.}
	\label{fig:norm-energy-alpha-init}
\end{figure} 
From Figure~\ref{fig:initializtion-comparison} we conclude that ZR gives the best approximation to the ground state. This is because, under the ZR setting the initial parameters of QAOA always starts corresponding to the exact ground state of $H_\text{mix}$ while the $H_\text{obj}$ is turned off. This is within the spirit of the homotopy optimization, in which starting in the optimal solution of the initial system is critical. Hence this good approximation to the initial parameters lead us to the better solution to the ground state of $H_\text{obj}$. Keeping in mind that if we sample $\alpha_\text{init}$ in the range $0\leq\alpha_\text{init}\leq 0.2$, we see that {NZR} shows comparable performance to {ZR} and the choice of initialization of $\gamma_j, \beta_j$ can be either one of them, relaxing the conditions on the choice of $\vec \gamma$ and $\vec \beta$. In the remaining of this paper all the numerical results are initialized with the ZR setting.

Now we move to the analysis of the choice of suitable $\alpha_\text{init}$. In the Figure~[\ref{fig:norm-energy-alpha-init}] we investigate the $\alpha_\text{init}$ dependency of the $E^{\ast}_\text{norm}$, where the energy is averaged over $100$ experiments. From Figure~[\ref{fig:norm-energy-alpha-init}](a) we take $3$ layers of HOHo-QAOA and observe that the mean optimal energy and the corresponding standard deviation remain unchanged (which we term as \textit{region of stability}) with respect to $\alpha_\text{init}$ in the range $0\leq\alpha_\text{init}\leq0.5$. With an increase in the number of nodes from $6$ to $16$, the \textit{region of stability} shifts upwards but remains in the range $0\leq\alpha_\text{init}\leq0.5$. This observation lead us to conclude that $\alpha_\text{init}$ can be chosen detached from zero without degrading the performance of HOHo-QAOA, or that at least that the region of stability does not shrink rapidly with the increased size of the problem. So setting $\alpha_\text{init}$ in the {region of stability} along with $\gamma_j^\text{init} = 0,\beta_j^\text{init} \sim \text{U}(0,2\pi)$ yields a solution with particularly small energy value.

In Figure~[\ref{fig:norm-energy-alpha-init}](b) we investigate how the efficiency of the optimization depends on the $\alpha_\text{step}$. During this investigation, we take $10$ layers of HOHo-QAOA. We observe that in the range $10^{-4}\leq\alpha_\text{step}<0.5$ the approximation to the ground energy and the corresponding standard deviation with increasing  $\alpha_\text{step}\rightarrow0$ remains almost unchanged, giving rise to  a {region of stability} with respect to $\alpha_\text{step}$. This behavior of $E_\text{norm}^\ast$ with $\alpha_\text{step}$ is similar to what we can observe for  $\alpha_\text{init}$. This lead us to a conclusion that one can choose $\alpha_\text{step}$ detached from zero for HOHo-QAOA. It should be noted that due to high simulation cost the experiment for $16$ qubits is halted at the $\alpha_\text{step} = 10^{-2}$, whereas the investigation for $6,\;16$ qubits is extended to $10^{-4}$.

The discussion and numerical results from the previous paragraphs give us the following initialization rules of HOHo-QAOA, which leads to a high efficiency of the method:
\begin{enumerate}
\item The parameters of mixer and objective should be initialized with {ZR} setting i.e. $\gamma_j^\text{init} = 0,\beta_j^\text{init} \sim \text{U}(0,2\pi)$,

\item Although one can infer that  $\alpha_\text{init}\rightarrow0$ along with $\alpha_\text{step}\rightarrow0$ gives the best result, our investigations show that one can choose the homotopy parameters detached from zero. This greatly reduces the cost of simulating HOHo-QAOA 
\end{enumerate}
\begin{figure}[t!]
	\centering
	\includegraphics{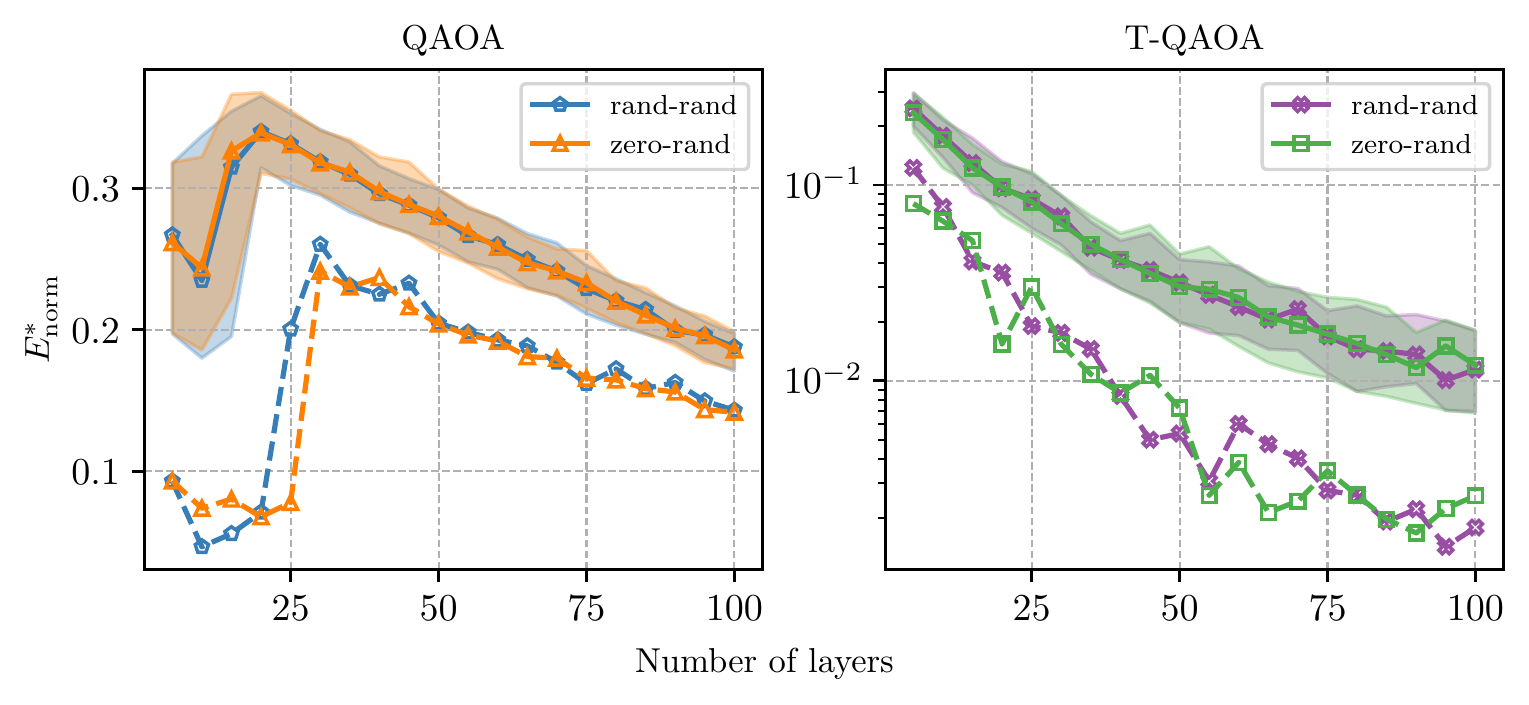}
	\caption{Comparison of different initialization for QAOA and T-QAOA. In the left (right) figure we illustrate how the $E^ *_\text{norm}$ changes with increasing number of layers in QAOA (T-QAOA) under the RR and ZR settings. The solid line is the  median energy over $100$ experiments, meanwhile the dashed line represents the best sample, taken layer-wise and node-wise by choosing the minimum energy among all the experiments. The areas are delimited by the first and third quartile.}
	\label{fig:benchmark-qaoa_t-qaoa}
\end{figure}

\begin{figure}[t!]
	\centering
	\includegraphics{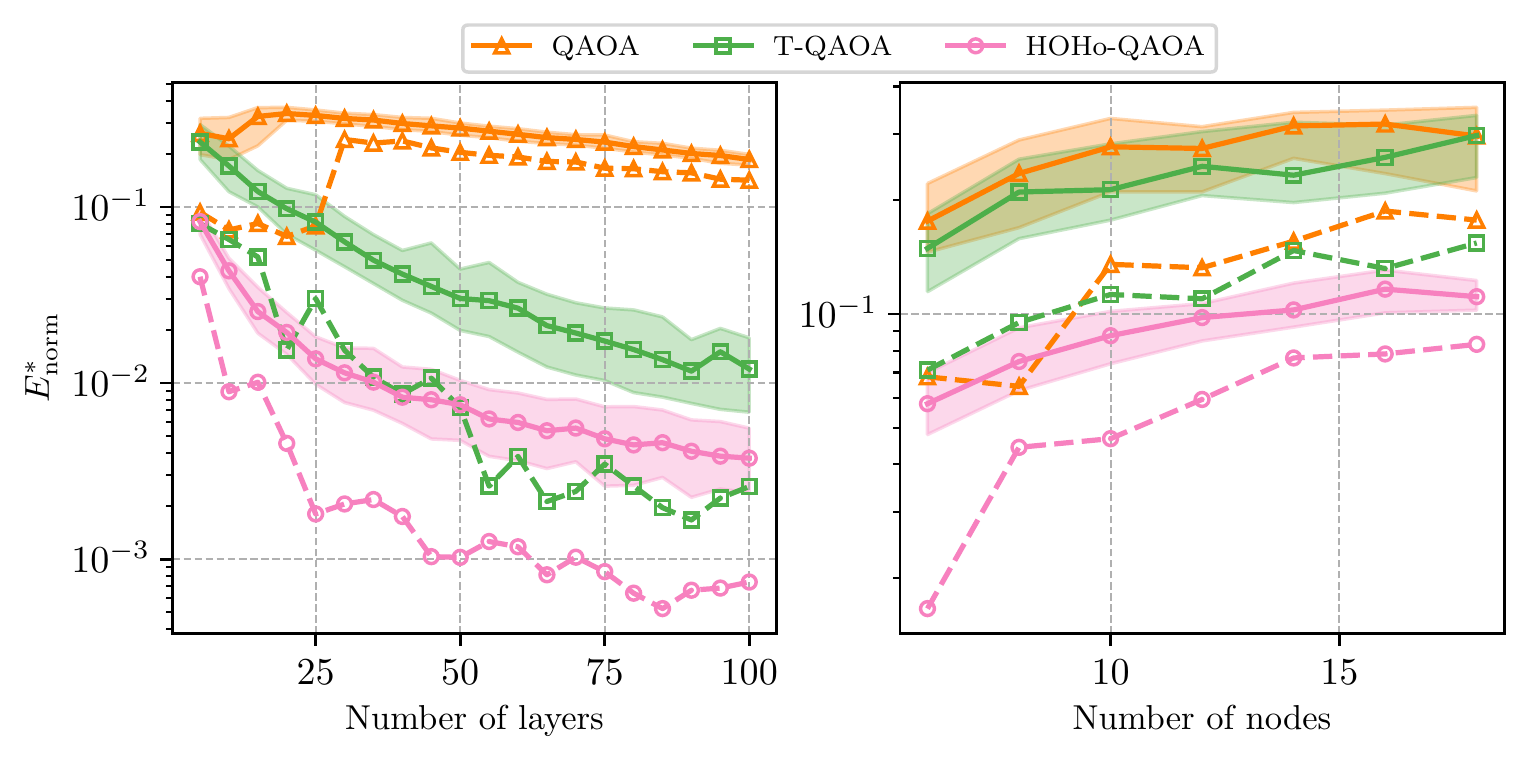}
	\caption{Performance of HOHo-QAOA compared to QAOA and T-QAOA. On both figures, for all the QAOA methods, we applied the ZR settings. The areas are delimited by the first and third quartile. The solid line presents $E^*_\text{norm}$ median over $100$ experiments for the left figure and $50$ experiments for the right figure, and the dashed line represents the best sample, taken layer-wise and node-wise by choosing the minimum energy among all the experiments.  On the left figure, the number of nodes is fixed to 10. On the right, the number of layers is fixed to 5 and the energy is sampled within 6 to 18 nodes. The homotopy parameters are set as $\alpha_\text{init}=0$ and $\alpha_\text{step}=0.01$. One can see that in both cases the averaged energy as well as the best sample of HOHo-QAOA outperforms the other variants of QAOA.  }
	\label{fig:benchmark_all}
\end{figure}

\subsection{Performance analysis}
In this section we  analyze the performance of the introduced algorithm with respect to other optimization strategies introduced above.
While it is natural for HOHo-QAOA to initialize using ZR strategy, 
it is unclear whether this choice will improve or worsen the results for QAOA or T-QAOA. Therefore before comparing state-of-the-art methods to the introduced one, we verify whether there is any difference in the performance for QAOA and T-QAOA with respect to the initialization of the optimized angles. In Fig.~\ref{fig:benchmark-qaoa_t-qaoa} we investigate state-of-the art methods for parameters $(\gamma_j, \beta_j)^\text{init}$ initialized with RR and ZR strategy. We observe that the performance of QAOA and T-QAOA is not influenced by the chosen strategies. This justifies using ZR strategy when comparing QAOA, T-QAOA and HOHo-QAOA.

Note that for QAOA we are observing undesired non-monotonic behavior with respect to the number of layers. We claim that this is caused because of a complicated landscape of the energy function, which makes difficult to optimize it if no information about the problem instance is used during the initialization from large number of nodes. This argument is complies with good performance of T-QAOA where the initial parameters of $(L+1)$-layer step is evaluated based on local optimal solutions of the $L$-layers step. 

In Fig.~\ref{fig:benchmark_all} we compare the performance of HOHo-QAOA with the other variants when $(\gamma_j, \beta_j)^\text{init}$ are initialized using ZR setting.
In the first experiment we run the algorithms with a fixed number of nodes while increasing number layers. In the second experiment the number of layers is fix while we vary the number of nodes. The plots present  optimized energy values, averaged respectively over $100$ and $50$ instances.
 The data shows that the introduced HOHo-QAOA gives us significantly smaller energy in both experiment setups. Good improvements remains  as more layers of the HOHo-QAOA are used and also outperforms the other varients of QAOA for higher number of nodes. This conclusions remain valid also for the best sample solution chosen (dashed line). It should be noted that the HOHo-QAOA outperforms QAOA and the T-QAOA in each and every layer starting from initial layer $5$ to final layer $100$.

\section{Conclusion}\label{sec:conclusion}
In the article we present a novel algorithm for combinatorial optimization. The method is a combination of homotopy optimization with an application in QAOA. In our method the observable used for computing the energy is changed during the optimization process. The process starts with observable being a mixer, for which the initial state of QAOA is a grounds state, and is slowly moved into the objective Hamiltonian. In addition we verify that, although traditionally in the homotopy method the initial value of transition parameter $\alpha$ should be 0 and the step should be as small as possible, for QAOA for the value of considered parameters can be detached from 0. 

A homotopy optimization is an algorithm dedicated for nonlinear optimized functions, and since even simple QAOA landscape is a linear combination of many -- for some problems exponentially many -- sinusoidal functions, our approach is well motivated for such energy function. This is in contrast to typical VQE optimization process, in which the function landscape with respect to a single parameter is just a sine. By comparing our approach and QAOA algorithm with the typical choice of optimization strategies we numerically confirmed that our method outperforms state-of-the-art approaches.

While our algorithm was only presented for QUBO and $X$-mixer,it is not restricted to it. In particular, if the transition function is of the form $H(\alpha) = g_1(\alpha)H_\text{mix}+g_2(\alpha) H_\text{obj}$, we only require energy of the $H_\textrm{mixer}$ to be efficiently computable. This includes XY-mixer~\cite{wang2020x} and Grover mixer~\cite{bartschi2020grover} for which the initial state can be efficiently prepared. Moreover, our approach remains also valid for higher-order binary problems~\cite{tabi2020quantum,glos2022space} and more advanced pseudo-code based QAOA Hamiltonian implementation~\cite{bako2022near}.

\paragraph{Acknowledgment}
$\hat[santa]{\text{A}}$.${\text{K}}$., $\hat[santa]{\text{A}}$.${\text{G}}$. and $\hat[santa]{\text{L}}$.${\text{B}}$. has been partially supported by Polish National Science Center under grant agreements 2019/33/B/ST6/02011. $\hat[santa]{\text{A}}$.${\text{G}}$. acknowledges support from National Science Center under grant agreement 2020/37/N/ST6/02220. The authors would like to thank Zolt\'an Zimbor\'as, \"Ozlem Salehi and Jaros\l aw A. Miszczak for valuable discussions and comments on the manuscript.

\paragraph{Data and code availability} Data and code available in \url{https://doi.org/10.5281/zenodo.7585691} 

\bibliographystyle{ieeetr}
\bibliography{reference}

\begin{thebibliography}{10}

\bibitem{preskill2018quantum}
J.~Preskill, ``Quantum computing in the {NISQ} era and beyond,'' {\em Quantum},
  vol.~2, p.~79, 2018.

\bibitem{bharti2022noisy}
K.~Bharti, A.~Cervera-Lierta, T.~H. Kyaw, T.~Haug, S.~Alperin-Lea, A.~Anand,
  M.~Degroote, H.~Heimonen, J.~S. Kottmann, T.~Menke, {\em et~al.}, ``Noisy
  intermediate-scale quantum algorithms,'' {\em Reviews of Modern Physics},
  vol.~94, no.~1, p.~015004, 2022.

\bibitem{cerezo2021variational}
M.~Cerezo, A.~Arrasmith, R.~Babbush, S.~C. Benjamin, S.~Endo, K.~Fujii, J.~R.
  McClean, K.~Mitarai, X.~Yuan, L.~Cincio, {\em et~al.}, ``Variational quantum
  algorithms,'' {\em Nature Reviews Physics}, vol.~3, no.~9, pp.~625--644,
  2021.

\bibitem{kandala2017hardware}
A.~Kandala, A.~Mezzacapo, K.~Temme, M.~Takita, M.~Brink, J.~M. Chow, and J.~M.
  Gambetta, ``Hardware-efficient variational quantum eigensolver for small
  molecules and quantum magnets,'' {\em Nature}, vol.~549, no.~7671,
  pp.~242--246, 2017.

\bibitem{bravo2019variational}
C.~Bravo-Prieto, R.~LaRose, M.~Cerezo, Y.~Subasi, L.~Cincio, and P.~J. Coles,
  ``Variational quantum linear solver,'' {\em arXiv preprint arXiv:1909.05820},
  2019.

\bibitem{lubasch2020variational}
M.~Lubasch, J.~Joo, P.~Moinier, M.~Kiffner, and D.~Jaksch, ``Variational
  quantum algorithms for nonlinear problems,'' {\em Physical Review A},
  vol.~101, no.~1, p.~010301, 2020.

\bibitem{larose2019variational}
R.~LaRose, A.~Tikku, {\'E}.~O’Neel-Judy, L.~Cincio, and P.~J. Coles,
  ``Variational quantum state diagonalization,'' {\em npj Quantum Information},
  vol.~5, no.~1, pp.~1--10, 2019.

\bibitem{kundu2022variational}
A.~Kundu and J.~A. Miszczak, ``Variational certification of quantum devices,''
  {\em Quantum Science and Technology}, vol.~7, no.~4, p.~045017, 2022.

\bibitem{farhi2014quantum}
E.~Farhi, J.~Goldstone, and S.~Gutmann, ``A quantum approximate optimization
  algorithm,'' {\em arXiv preprint arXiv:1411.4028}, 2014.

\bibitem{farhi2001quantum}
E.~Farhi, J.~Goldstone, S.~Gutmann, J.~Lapan, A.~Lundgren, and D.~Preda, ``A
  quantum adiabatic evolution algorithm applied to random instances of an
  np-complete problem,'' {\em Science}, vol.~292, no.~5516, pp.~472--475, 2001.

\bibitem{tabi2020quantum}
Z.~Tabi, K.~H. El-Safty, Z.~Kallus, P.~H{\'a}ga, T.~Kozsik, A.~Glos, and
  Z.~Zimbor{\'a}s, ``Quantum optimization for the graph coloring problem with
  space-efficient embedding,'' in {\em 2020 IEEE International Conference on
  Quantum Computing and Engineering (QCE)}, pp.~56--62, IEEE, 2020.

\bibitem{bako2022near}
B.~Bak{\'o}, A.~Glos, {\"O}.~Salehi, and Z.~Zimbor{\'a}s, ``Near-optimal
  circuit design for variational quantum optimization,'' {\em arXiv preprint
  arXiv:2209.03386}, 2022.

\bibitem{farhi2014maxkcut}
E.~Farhi, J.~Goldstone, and S.~Gutmann, ``A quantum approximate optimization
  algorithm applied to a bounded occurrence constraint problem,'' {\em arXiv
  preprint arXiv:1412.6062}.

\bibitem{cook2020quantum}
J.~Cook, S.~Eidenbenz, and A.~B{\"a}rtschi, ``The quantum alternating operator
  ansatz on maximum k-vertex cover,'' in {\em 2020 IEEE International
  Conference on Quantum Computing and Engineering (QCE)}, pp.~83--92, IEEE,
  2020.

\bibitem{glos2022space}
A.~Glos, A.~Krawiec, and Z.~Zimbor{\'a}s, ``Space-efficient binary optimization
  for variational quantum computing,'' {\em npj Quantum Information}, vol.~8,
  no.~1, pp.~1--8, 2022.

\bibitem{medvidovic2021classical}
M.~Medvidovi{\'c} and G.~Carleo, ``Classical variational simulation of the
  quantum approximate optimization algorithm,'' {\em npj Quantum Information},
  vol.~7, no.~1, pp.~1--7, 2021.

\bibitem{wang2018quantum}
Z.~Wang, S.~Hadfield, Z.~Jiang, and E.~G. Rieffel, ``{Quantum approximate
  optimization algorithm for MaxCut: A fermionic view},'' {\em Physical Review
  A}, vol.~97, no.~2, p.~022304, 2018.

\bibitem{alam2020accelerating}
M.~Alam, A.~Ash-Saki, and S.~Ghosh, ``Accelerating quantum approximate
  optimization algorithm using machine learning,'' in {\em 2020 Design,
  Automation \& Test in Europe Conference \& Exhibition (DATE)}, pp.~686--689,
  IEEE, 2020.

\bibitem{shaydulin2019multistart}
R.~Shaydulin, I.~Safro, and J.~Larson, ``Multistart methods for quantum
  approximate optimization,'' in {\em 2019 IEEE high performance extreme
  computing conference (HPEC)}, pp.~1--8, IEEE, 2019.

\bibitem{hegade2021portfolio}
N.~Hegade, P.~Chandarana, K.~Paul, X.~Chen, F.~Albarr{\'a}n-Arriagada, and
  E.~Solano, ``Portfolio optimization with digitized-counterdiabatic quantum
  algorithms,'' {\em arXiv preprint arXiv:2112.08347}, 2021.

\bibitem{zhu2022adaptive}
L.~Zhu, H.~L. Tang, G.~S. Barron, F.~Calderon-Vargas, N.~J. Mayhall, E.~Barnes,
  and S.~E. Economou, ``Adaptive quantum approximate optimization algorithm for
  solving combinatorial problems on a quantum computer,'' {\em Physical Review
  Research}, vol.~4, no.~3, p.~033029, 2022.

\bibitem{zhouprx2020}
L.~Zhou, S.-T. Wang, S.~Choi, H.~Pichler, and M.~D. Lukin, ``Quantum
  approximate optimization algorithm: Performance, mechanism, and
  implementation on near-term devices,'' {\em Physical Review X}, vol.~10,
  p.~021067, Jun 2020.

\bibitem{zhou2022qaoa}
Z.~Zhou, Y.~Du, X.~Tian, and D.~Tao, ``{QAOA-in-QAOA: solving large-scale
  MaxCut problems on small quantum machines},'' {\em arXiv preprint
  arXiv:2205.11762}, 2022.

\bibitem{patel2022rlrqaoa}
Y.~J. Patel, S.~Jerbi, T.~Bäck, and V.~Dunjko, ``{Reinforcement Learning
  Assisted Recursive QAOA},'' {\em arXiv preprint arXiv.2207.06294}, 2022.

\bibitem{layne1989optimization}
L.~T. Watson and R.~T. Haftka, ``Modern homotopy methods in optimization,''
  {\em Computer Methods in Applied Mechanics and Engineering}, vol.~74, no.~3,
  pp.~289--305, 1989.

\bibitem{garcia2018addressing}
A.~Garcia-Saez and J.~Latorre, ``Addressing hard classical problems with
  adiabatically assisted variational quantum eigensolvers,'' {\em arXiv
  preprint arXiv:1806.02287}, 2018.

\bibitem{mcclean2016theory}
J.~R. McClean, J.~Romero, R.~Babbush, and A.~Aspuru-Guzik, ``The theory of
  variational hybrid quantum-classical algorithms,'' {\em New Journal of
  Physics}, vol.~18, no.~2, p.~023023, 2016.

\bibitem{harwood2022improving}
S.~M. Harwood, D.~Trenev, S.~T. Stober, P.~Barkoutsos, T.~P. Gujarati,
  S.~Mostame, and D.~Greenberg, ``Improving the variational quantum eigensolver
  using variational adiabatic quantum computing,'' {\em ACM Transactions on
  Quantum Computing}, vol.~3, no.~1, pp.~1--20, 2022.

\bibitem{mcclean2021low}
J.~R. McClean, M.~P. Harrigan, M.~Mohseni, N.~C. Rubin, Z.~Jiang, S.~Boixo,
  V.~N. Smelyanskiy, R.~Babbush, and H.~Neven, ``Low-depth mechanisms for
  quantum optimization,'' {\em PRX Quantum}, vol.~2, no.~3, p.~030312, 2021.

\bibitem{messiah1961quantum}
A.~Messiah and G.~Temmer, {\em Quantum Mechanics}.
\newblock No.~v. 1 in Quantum Mechanics, North-Holland Publishing Company,
  1961.

\bibitem{albash2018adiabatic}
T.~Albash and D.~A. Lidar, ``Adiabatic quantum computation,'' {\em Reviews of
  Modern Physics}, vol.~90, no.~1, p.~015002, 2018.

\bibitem{marzlin2004inconsistency}
K.-P. Marzlin and B.~C. Sanders, ``Inconsistency in the application of the
  adiabatic theorem,'' {\em {Physical Review Letters}}, vol.~93, no.~16,
  p.~160408, 2004.

\bibitem{tong2005quantitative}
D.~Tong, K.~Singh, L.~C. Kwek, and C.~H. Oh, ``Quantitative conditions do not
  guarantee the validity of the adiabatic approximation,'' {\em {Physical
  Review Letters}}, vol.~95, no.~11, p.~110407, 2005.

\bibitem{du2008experimental}
J.~Du, L.~Hu, Y.~Wang, J.~Wu, M.~Zhao, and D.~Suter, ``Experimental study of
  the validity of quantitative conditions in the quantum adiabatic theorem,''
  {\em {Physical Review Letters}}, vol.~101, no.~6, p.~060403, 2008.

\bibitem{wu2008adiabatic}
J.-d. Wu, M.-s. Zhao, J.-l. Chen, and Y.-d. Zhang, ``Adiabatic condition and
  quantum geometric potential,'' {\em Physical Review A}, vol.~77, p.~062114,
  2008.

\bibitem{salehi2022unconstrained}
{\"O}.~Salehi, A.~Glos, and J.~A. Miszczak, ``Unconstrained binary models of
  the travelling salesman problem variants for quantum optimization,'' {\em
  Quantum Information Processing}, vol.~21, no.~2, pp.~1--30, 2022.

\bibitem{domino2021quantum}
K.~Domino, M.~Koniorczyk, K.~Krawiec, K.~Ja{\l}owiecki, S.~Deffner, and
  B.~Gardas, ``Quantum annealing in the {NISQ} era: railway conflict
  management,'' {\em arXiv preprint arXiv:2112.03674}, 2021.

\bibitem{domino2022quadratic}
K.~Domino, A.~Kundu, {\"O}.~Salehi, and K.~Krawiec, ``Quadratic and
  higher-order unconstrained binary optimization of railway rescheduling for
  quantum computing,'' {\em Quantum Information Processing}, vol.~21, no.~9,
  pp.~1--33, 2022.

\bibitem{borowski2020new}
M.~Borowski, P.~Gora, K.~Karnas, M.~B{\l}ajda, K.~Kr{\'o}l, A.~Matyjasek,
  D.~Burczyk, M.~Szewczyk, and M.~Kutwin, ``New hybrid quantum annealing
  algorithms for solving vehicle routing problem,'' in {\em International
  Conference on Computational Science}, pp.~546--561, Springer, 2020.

\bibitem{Arya2022}
A.~Arya, L.~Botelho, F.~Ca{\~{n}}ete, D.~Kapadia, and {\"O}.~Salehi, {\em
  Applications of Quantum Annealing to Music Theory}, pp.~373--406.
\newblock Cham: Springer International Publishing, 2022.

\bibitem{glos2022testvehicle}
A.~Glos, A.~Kundu, and {\"O}.~Salehi, ``Optimizing the production of test
  vehicles using hybrid constrained quantum annealing,'' {\em arXiv preprint
  arXiv:2203.15421}, 2022.

\bibitem{babbush2014adiabatic}
R.~Babbush, P.~J. Love, and A.~Aspuru-Guzik, ``Adiabatic quantum simulation of
  quantum chemistry,'' {\em {Scientific Reports}}, vol.~4, no.~1, pp.~1--11,
  2014.

\bibitem{dascolloquium-qa-aqc}
A.~Das and B.~K. Chakrabarti, ``Colloquium: Quantum annealing and analog
  quantum computation,'' {\em Review of Modern Physics}, vol.~80,
  pp.~1061--1081, Sep 2008.

\bibitem{hadfield2019quantum}
S.~Hadfield, Z.~Wang, B.~O’gorman, E.~G. Rieffel, D.~Venturelli, and
  R.~Biswas, ``From the quantum approximate optimization algorithm to a quantum
  alternating operator ansatz,'' {\em Algorithms}, vol.~12, no.~2, p.~34, 2019.

\bibitem{bartschi2020grover}
A.~B{\"a}rtschi and S.~Eidenbenz, ``Grover mixers for {QAOA}: Shifting
  complexity from mixer design to state preparation,'' in {\em 2020 IEEE
  International Conference on Quantum Computing and Engineering (QCE)},
  pp.~72--82, IEEE, 2020.

\bibitem{wang2020x}
Z.~Wang, N.~C. Rubin, J.~M. Dominy, and E.~G. Rieffel, ``{XY}-mixers:
  Analytical and numerical results for the quantum alternating operator
  ansatz,'' {\em Physical Review A}, vol.~101, no.~1, p.~012320, 2020.

\bibitem{crosson2014different}
E.~Crosson, E.~Farhi, C.~Y.-Y. Lin, H.-H. Lin, and P.~Shor, ``Different
  strategies for optimization using the quantum adiabatic algorithm,'' {\em
  arXiv preprint arXiv:1401.7320}, 2014.

\bibitem{PhysRevX.6.031010}
S.~Muthukrishnan, T.~Albash, and D.~A. Lidar, ``Tunneling and speedup in
  quantum optimization for permutation-symmetric problems,'' {\em Physical
  Review X}, vol.~6, p.~031010, Jul 2016.

\bibitem{RevModPhys.90.015002}
T.~Albash and D.~A. Lidar, ``Adiabatic quantum computation,'' {\em Review of
  Modern Physics}, vol.~90, p.~015002, Jan 2018.

\bibitem{ostaszewski2021structure}
M.~Ostaszewski, E.~Grant, and M.~Benedetti, ``Structure optimization for
  parameterized quantum circuits,'' {\em Quantum}, vol.~5, p.~391, 2021.

\bibitem{nakanishi2020sequential}
K.~M. Nakanishi, K.~Fujii, and S.~Todo, ``Sequential minimal optimization for
  quantum-classical hybrid algorithms,'' {\em Physical Review Research},
  vol.~2, no.~4, p.~043158, 2020.

\end{thebibliography}

\appendix
\section{Experiment details}\label{app:ex-details}

In order to enable the simple reproduction of our results, we publish our code on ... The algorithms for generating data and plotting were implemented in Julia and Python programming languages. Versions of the software and additional packages are listed in... 

\paragraph{Experiments} Each experiment of HOHo-QAOA is uniquely characterized by the random graph $G=(V,E)$, that is chosen from Barab\'asi-Albert distribution with $6, 8,\ldots 18$ nodes and with $m=2$, where $m$ defines the number of edges to be attached from a new node to existing ones. The weights corresponding to the edges are picked up from a uniform set of integer weights $w_{jj^\prime}\in\{1,\ldots,10\}$ for each edge $\{j,j^\prime\}$.

\paragraph{Data Sampling} For sampling the objective Hamiltonian, we started by generating graph objects an them converting to Pauli operators objects and Hamiltonian matrices with \texttt{Qiskit}. We generated $100$ samples for each graph setup. We sampled  the initial optimization parameters in a random distribution for RR and ZR approaches.  We emulated the quantum evolution and take an exact expectation energy and gradient of the state during the optimization. We choose the \texttt{L-BFGS} algorithm implemented in Julia’s \texttt{Optim} package as a subroutine. The optimization has no periodic or bounds conditions. We setup Optim with absolute tolerance, relative tolerance and absolute tolerance in gradient equal to $1^{-9}$. We allowed steps that increase the objective value and maximum number of iterations is $10000$.

\paragraph{T-QAOA} For T-QAOA implementation, we initialize with a minimum number of levels $L_0 = 4$ and run the optimization similarly to the state of art QAOA with the a given parameters initiation strategy. The method proceeds checking the convergence of the solution and moving to the next layer $L_0+1$, using the previous optimized parameters with the addition of a zero for the mixer Hamiltonian and a value sampled from a uniform random distribution $\text{U}(0,2\pi)$.

\section{Proof of nonlinear landscape for QAOA}\label{apndx:proof-non-linear-lanscape}

\begin{theorem}
	Let $\varrho$ be an arbitrary quantum state, $H$ be an arbitrary Hamiltonian with spectrum set  $\{E_1,\dots,E_k\}$ and $O$ be an arbitrary observable. Then
	\begin{equation}
		\tr(\exp(-\ii \theta H)\varrho\exp(\ii \theta H) O ) =	C+ \sum_{i>j} A_{i,j}\cos(\theta(E_i-E_j) + B_{i,j}),
	\end{equation}
	for some real values $C,A_{i,j},B_{i,j}$.
\end{theorem}

\begin{proof}
	Let $U$  be a unitary that diagonalizes the Hamiltonian $H$. Then we have
	\begin{equation}
		\begin{split}
			\tr(\exp(-\ii \theta H)\varrho\exp(\ii \theta  H) O ) &= \tr \left( \sum_{i=1}^k (U e^{-\ii\theta E_i} \ketbra{i} U^\dagger )\varrho\sum_{j=1}^k (U e^{\ii \theta E_j} \ketbra {j}U^\dagger  )O\right)\\
			&= \sum_{i=1}^k \sum_{j=1}^k e^{\ii \theta (E_j- E_i)}\tr \left(    U \ketbra i U^\dagger \varrho U \ketbra j U^\dagger O\right)\\
			&= \sum_{i=1}^k \sum_{j=1}^k e^{\ii \theta (E_j- E_i)}\tr \left(     \ketbra i  \varrho' \ketbra j  O'\right)\\
			&=  \sum_{i=1}^k \sum_{j=1}^k e^{\ii \theta (E_j- E_i)}     \bra i  \varrho' \ket j \bra j  O' \ket i,
		\end{split}
	\end{equation}
	where $\varrho' = U^\dagger \varrho U$ and $O' = U^\dagger O U$. Since $\varrho'$ is a hermitian operator and therefore $\bra i \varrho \ket j = \overline{ \bra j \varrho  \ket i}$, and similarly for $O'$,  therefore for any $i, j$ the term for $i>j$ is a conjugate of the term $i<j$. Hence
	\begin{equation}
		\sum_{i=1}^k \sum_{j=1}^k e^{\ii \theta (E_j- E_i)}     \bra i  \varrho' \ket j \bra j  O' \ket i
		=  \sum_{i=1}^k   \bra i  \varrho' \ket i \bra i  O' \ket i   +2 \sum_{i>j}\Re   e^{\ii \theta (E_j- E_i)}     \bra i  \varrho' \ket j \bra j  O' \ket i.
	\end{equation}
	Note that the left hand side sum in the above above is a free term and is a real number. Starting from now we will assume that the Hamiltonian $H$ is non-degenerate -- otherwise the corresponding element of the right sum will contribute to the free term. Taking $x_{i,j} + \ii y_{i,j}\coloneqq \bra i  \varrho' \ket j \bra j  O' \ket i$  for some real $x_{i,j}, y_{i,j}$ we have
	\begin{multline}
		\Re   e^{\ii \theta (E_j- E_i)}     \bra i  \varrho' \ket j \bra j  O' \ket i
		= \Re (\cos(\theta (E_j- E_i)) + \ii \sin(\theta (E_j- E_i))) (x_{i,j} + \ii y_{i,j}) \\
		=x_{i,j}  \cos(\theta (E_j- E_i)) - y_{i,j} \sin(\theta (E_j- E_i)) \\
		 = \sqrt{x_{i,j}^2 + y_{i,j}^2} \left ( \frac{x_{i,j}}{\sqrt{x_{i,j}^2 + y_{i,j}^2}}\cos(\theta (E_j- E_i))  - \frac{y_{i,j} }{\sqrt{x_{i,j}^2 + y_{i,j}^2}} \sin(\theta (E_j- E_i))  \right ) \\
		= \sqrt{x_{i,j}^2 + y_{i,j}^2} \left(\cos(\alpha_{i,j})  \cos(\theta (E_j- E_i)) - \sin (\alpha_{i,j}) \sin(\theta (E_j- E_i))  \right),
	\end{multline}
	where $\alpha_{i,j}$ is such a real number for which the above transformation holds. Note that such a number $\alpha$ can always be found as the replaced fraction squared sum to 1 and one can use Pythagorean trigonometric identity. Finally we have
	\begin{multline}
		\sqrt{x_{i,j}^2 + y_{i,j}^2} \left (\cos(\alpha_{i,j})  \cos(\theta (E_j- E_i)) - \sin (\alpha_{i,j}) \sin(\theta (E_j- E_i))  \right )\\
		= \sqrt{x_{i,j}^2 + y_{i,j}^2} \cos(  \theta(E_j- E_i) + \alpha_{i,j} ),
	\end{multline}
	which proves the statement of the theorem.
\end{proof}
Note that the case of Hamiltonian with two different eigenvalues was already presented in \cite{ostaszewski2021structure, nakanishi2020sequential}.

\end{document}